\documentclass[11pt,A4paper,fullpage]{article}

\usepackage{todonotes} 

\usepackage{amsfonts}
\usepackage{fullpage}
\usepackage{graphicx,color}
\usepackage{xspace}
\usepackage{enumerate}
\usepackage{amssymb,amsmath}
\usepackage{amsthm}
\usepackage{subfig}

\usepackage[utf8]{inputenc}
\usepackage{times}

\usepackage[utf8]{inputenc}
\usepackage{authblk}
\usepackage{hyperref}

\usepackage{mathtools}

\newtheorem{theorem}{Theorem}

\newtheorem{corollary}[theorem]{Corollary}
\newtheorem{proposition}[theorem]{Proposition}
\newtheorem{lemma}[theorem]{Lemma}

\newtheorem{remark}[theorem]{Remark}

\newtheorem{property}[theorem]{Property}

\newtheorem{definition}[theorem]{Definition}

\newcommand{\mch}[2]{
\left.\mathchoice
  {\left(\kern-0.48em\binom{#1}{#2}\kern-0.48em\right)}
  {\big(\kern-0.30em\binom{\smash{#1}}{\smash{#2}}\kern-0.30em\big)}
  {\left(\kern-0.30em\binom{\smash{#1}}{\smash{#2}}\kern-0.30em\right)}
  {\left(\kern-0.30em\binom{\smash{#1}}{\smash{#2}}\kern-0.30em\right)}
\right.}

\newcommand{\wa}{{\sc Weighting}\xspace}

\newcommand{\acc}{\ensuremath{\textrm{acc}}\xspace}

\newcommand{\search}{\textsc{Search}$(n,k,H)$\xspace}

\begin{document}
\title{Competitive Sequencing with Noisy Advice}

\author[1]{Spyros Angelopoulos\thanks {Corresponding author: {email: {\tt spyros.angelopoulos@lip6.fr}}}}
\author[2]{Diogo Ars\'enio} 
\author[3]{Shahin Kamali}

\date{}

\affil[1]{CNRS and Sorbonne University, Paris, France.
} 
\medskip

\affil[2]{New York University Abu Dhabi, Abu Dhabi, UAE. 
}

\affil[3]{University of Manitoba, Winnipeg, Canada. 
}

\maketitle

\begin{abstract}
Several well-studied online resource allocation problems can be formulated in terms of infinite, increasing sequences of positive values, in which each element is associated with a corresponding allocation value. Examples include problems such as online bidding, searching for a hidden target on an unbounded line, and designing interruptible algorithms based on repeated executions. The performance of the online algorithm, in each of these problems, is measured by the competitive ratio, which describes the multiplicative performance loss due to the absence of full information on the instance.

We study such competitive sequencing problems in a setting in which the online algorithm has some (potentially) erroneous information, 
expressed as a $k$-bit advice string, for some given $k$. 
We first consider the {\em untrusted} advice setting of [Angelopoulos et al, ITCS 2020], in which the objective is to quantify performance considering two extremes: either the advice is either error-free, or it is generated by a 
(malicious) adversary. Here, we show a Pareto-optimal solution, using a new approach for applying the functional-based lower-bound technique due to [Gal, Israel J. Math. 1972]. Next, we study a nascent {\em noisy} advice setting, in which a number of the advice bits may be erroneous; the exact error is unknown to the online algorithm, which only has access to a pessimistic estimate (i.e., an upper bound on this error). We give improved upper bounds, but also the first lower bound on the competitive ratio of an online problem in this setting. To this end, we combine ideas from robust query-based search in arrays, and fault-tolerant contract scheduling. Last, we demonstrate how to apply the above techniques in robust optimization without predictions, and discuss how they can be applicable in the context of more general online problems.
\end{abstract}

\newpage

\section{Introduction}
\label{sec:introduction}

Online computation, and competitive analysis, in particular, has served as the definitive framework for the theoretical analysis of algorithms in a state of uncertainty. While the early, standard definition of online computation~\cite{SleTar85} assumes that the algorithm has no knowledge in regards to the request sequence, in practical situations the algorithm will indeed have certain limited, but possibly inaccurate such information (e.g., some lookahead, or historical information on typical sequences). Hence the need for more nuanced models that capture the power and limitations of online algorithms enhanced with external information. 

One such approach, within Theoretical Computer Science (TCS), is the framework of {\em advice complexity}; see~\cite{DobKraPar09,BocKomKra09,EmekFraKorRos2011}, the survey~\cite{Boyar:survey:2016} and the book~\cite{komm2016introduction}. In the advice-complexity model (and in particular, the {\em tape} model~\cite{DBLP:journals/jcss/BockenhauerKKK17} of advice complexity), the online algorithm receives a string that encodes information concerning the request sequence, and which can help improve its performance.
The objective is to quantify the trade-offs between the size of advice (in terms of number of bits), and the competitive ratio of the algorithm. This model places stringent requirements: the advice is assumed to be {\em error free}, and may be provided by an {\em omnipotent oracle}. As such, this model is mostly of theoretical significance, with very limited applications in practice.

To address these limitations,~\cite{DBLP:conf/innovations/0001DJKR20} introduced a model in which the advice can be erroneous, termed the {\em untrusted advice} model. Here, the algorithm's performance is evaluated at two extreme situations, in regards to the advice error. At the one extreme, the advice is error-free, in which case the competitive ratio is termed the {\em consistency} of the algorithm. At the other extreme,  the advice is generated by a (malicious) adversary who aims to maximize the performance degradation of the algorithm, in which case the competitive ratio is termed the {\em robustness} of the algorithm. The objective is to identify strategies that are {\em Pareto-efficient}, and ideally {\em Pareto-optimal}, for these two extreme measures.  Several online problems have been studied recently within this framework of Pareto-optimality; see, e.g.,~\cite{wei2020optimal,li2021robustness,lee2021online,DBLP:conf/innovations/000121, DBLP:conf/aaai/0001K21}. Note also that the definitions of consistency and robustness are borrowed from recent advances on {\em learning-augmented} online algorithms and online algorithms with {\em predictions}, as shown first in~\cite{DBLP:conf/icml/LykourisV18} and~\cite{NIPS2018_8174}. See also the survey~\cite{mitzenmacher2020algorithms}.

\subsection{Online computation with noisy advice}

In this work we focus on a nascent model of advice, in which the advice string can be {\em noisy}. The starting observation is that the Pareto-based framework of untrusted advice only focuses on the tradeoff between the extreme competitive ratios, namely the consistency and the robustness. A more general question is to evaluate the performance of an online algorithm under the assumption that only up to a certain number of the advice bits may be erroneous. Given an advice string of size $k$, we denote by $\eta\leq k$ the number of 
erroneous bits. The objective is to study the power and limitations of online algorithms within this setting, i.e., from the point of view of both upper and lower bounds on the competitive ratios.

Naturally, the algorithm does not know $\eta$. We will assume, however, that the algorithm has knowledge of an {\em upper bound} $H$ on the error, i.e., the algorithm knows that $\eta\leq H\leq k$.
We further distinguish between two possible models: The first, and perhaps more natural model, treats $H$ as an ``absolute'' upper bound on the error, in the sense that there is no requirement that the algorithm performs well if the error happens to exceed $H$. We call this 
the {\em standard} model of noisy advice, and we often omit the word standard when referring to it. The second model treats $H$ as an ``anticipated'' upper bound, and enforces the requirement that the online algorithm performs well not only in the case 
$\eta\leq H$, but should also satisfy some performance guarantees if it so happens that $\eta>H$ (and in particular, if $\eta$ is adversarial).
The objective here is to quantify the tradeoff between these two objectives. 
We call this model the {\em robust noisy} model. As we will discuss, our analysis allows us to obtain results for both models. 

Note that the noisy advice becomes even more pertinent if each advice bit is a response to a {\em binary query} concerning the input. While the noisy model, as defined above, does not impose such a restriction on the nature of advice, this aspect can be useful in realistic settings in which there is a prediction in the form of response to $k$ binary queries. To our knowledge, the noisy advice model (in particular, its binary query-based interpretation) has only been applied to the problem of {\em contract scheduling}~\cite{DBLP:conf/aaai/0001K21}, and solely from the point of view of upper bounds. This problem belongs in a larger class of problems, which we discuss below.

\subsection{Competitive sequencing and applications}
\label{subsec:contract.scheduling}

Consider the following simple problem: Define a sequence\footnote{Since we will often use modulo arithmetic in this work, we will use indexing starting with ``0'' instead of ``1''. } of increasing, positive numbers  
$X=(x_i)_{i=0}^\infty$ which minimizes the quantity
\begin{equation}
\rho_X = \sup_{i \geq 0} \frac{\sum_{j=0}^i x_j}{x_{i-1}},
\label{eq:bidding}
\end{equation}
where $x_{-1}$ is defined to be equal to 1. 

In TCS, this problem is known as {\em online bidding}~\cite{ChrKen06} and is used as a vehicle for formalizing efficient doubling strategies for certain online and offline optimization problems. Here, each $x_i$ can be thought as a ``bid''; there is also a hidden (unknown) value $u$. The algorithm must define a sequence of bids, which minimizes the worst-case ratio of the cost of the algorithm over the value $u$, where the cost is defined as the sum of bids in the sequence up and including the first bid that is at least as high as $u$. Note that this worst-case ratio is maximized if $u$ is chosen to be infinitesimally larger than a bid, hence the worst-case ratio is given precisely by~\eqref{eq:bidding}, which is called the {\em competitive ratio} of the bidding strategy. The problem is also related to {\em searching on the line} (i.e., on two concurrent branches with a common origin) for a hidden target~\cite{bellman,beck:yet.more,yates:plane}, a problem with a very rich history of study in TCS and Operations Research. More precisely, $1+2\rho_X$ is the competitive ratio of a search strategy based on the sequence $X$, in which, in round $i$, the searcher explores the branch 
$i \mod 2$ of the line up to length equal to $x_i$.

In AI, competitive sequencing is the formulation of another important resource allocation problem known as {\em contract scheduling}~\cite{RZ.1991.composing}. Here, the objective is to allocate computational time to an algorithm so that it outputs a reasonably efficient solution, even if interrupted during its execution. More precisely, let $A$ be an algorithm which may return a meaningless solution if interrupted at some point prior to the end of its execution (e.g., a PTAS based on dynamic programming); such algorithms are known as {\em contract algorithms}. We would like to obtain an algorithm based on $A$ that is, instead, robust to interruptions, and to this end, we repeatedly execute $A$ with increasing runtimes, as dictated by the sequence $X$. In other words, $x_i$ is the runtime (or {\em length}) of the $i$-th execution in what we call a {\em schedule} of (executions) of contract algorithms. 
Suppose that an interruption occurs at time $T$, and let $\ell(X,T)$
denote the longest execution that completed by time $T$. The worst-case ratio $\sup_T T/\ell(X,T)$ is known as the {\em acceleration ratio} of the schedule that is based on the sequence $X$~\cite{RZ.1991.composing}. Since the acceleration ratio is maximized for interruptions that occur infinitesimally prior to the completion of a contract algorithm, it follows that the acceleration ratio is precisely equal to 
$\rho_X$, as expressed in~\eqref{eq:bidding}. 

For all the above problems, optimizing $\rho_X$ implies optimizing their corresponding, worst-case performance guarantee. It has long been known that the doubling sequence $X=(2^i)_{i=0}^\infty$ is optimal, and $\rho_X=4$ is best-possible. Beyond these simple settings, some of the above problems, and most notably contract scheduling, have been studied under more complex settings, for which optimizing the corresponding performance measure becomes much more challenging. This includes settings in which one must schedule executions of the contract algorithm for several problem instances, as well as the setting in which the executions can be scheduled in several parallel processors, see, e.g.,~\cite{BPZF.2002.scheduling,ZilbersteinCC03,steins,aaai06:contracts,soft-contracts,ALO:multiproblem,DBLP:conf/ijcai/0001J19,kupavskii2018lower}.
All these works show theoretical upper and lower bounds on the performance of the corresponding schedules.

To our knowledge, the only previous work on noisy advice for an online problem is in the context of contract scheduling. More specifically,~\cite{DBLP:conf/aaai/0001K21} gave a schedule for the robust noisy model which uses the responses to the $k$ binary queries so as to select the best schedule from a space of $k$ appropriately defined schedules. The same work also gave a schedule is that is Pareto-optimal for the untrusted advice model, only for the case that we require that the robustness is equal to 4.

\subsection{Contributions}
\label{subsec:contributions}

In this work we study the power, but also the limitations of algorithms for competitive sequencing problems with noisy advice. In particular, we give new approaches that allow us to show the first lower bounds on the competitive ratio of algorithms in this setting. For concreteness, we will follow the abstraction of the contract scheduling problem, since it has the most intuitive statement among sequencing problems, but also has a significant body of precious work. We emphasize that our overall approach and techniques should carry over to other online problems mentioned in Section~\ref{subsec:contract.scheduling}. 

 We begin with the untrusted advice model: here, we show a tight lower bound that establishes the Pareto optimality of the schedule given in~\cite{DBLP:conf/aaai/0001K21} for all possible values of robustness $r\geq 4$, and any size $k$ of the advice string (Section~\ref{sec:tradeoff}). This also helps us establish properties that will be useful in the study of the noisy model. To achieve this result, we introduce a new application of Gal's {\em functional theorem}~\cite{gal:general}. This is the main tool for showing lower bounds in search theory but also in contract scheduling: see, e.g., Chapters 7 and 9 in~\cite{searchgames} and~\cite{alex:robots,schuierer:lower}, for its applications in search theory, as well as~\cite{aaai06:contracts,soft-contracts,ALO:multiproblem,DBLP:conf/ijcai/0001J19} for its applications in contract scheduling. Note that all previous work has applied this theorem in a ``black-box'' manner: specifically, a typical application of this theorem gives a lower bound to the performance of the algorithms, as a function of a parameter  $\alpha_X$, that depends on the sequence $X$ (see Theorem~\ref{thm:gal} for details); say that $f(\alpha_X)$ is this lower bound expression. The next step is then to find a specific solution (usually a simple, geometric sequence of the form $(b^i)_{i=0}^\infty$, for some chosen $b>1$) with competitive ratio at most $f(b)$, where $f$ must crucially be the same function as the one that gives the lower bound. Choosing $b$ so as to minimize $f(b)$ yields an optimal result, and note that neither the specific range of $\alpha_X$, nor $\alpha_X$ itself are, in a sense, significant in this standard approach. In our work, instead, due to the inherent complexity of the solution $X$ (since it is determined by the advice bits), the above approach does not apply directly. To bypass this difficulty, we first give a lower bound of the consistency of any $r$-robust schedule with $k$ advice bits by treating it a ``virtual'' multi-processor problem using $2^k$ searchers, and by defining a suitable parameter $\alpha_X$ that makes the application of Gal's theorem possible. But more importantly, one needs to have a close look at this parameter, and limit the range of $\alpha_X$, in order to obtain a tight bound. The crucial part is then to give upper and lower bounds on $\alpha_X$, in terms of $k$ and $r$ 
(see Theorem~\ref{thm:zetas} and Corollary~\ref{cor:merge}).

Our next contribution concerns the setting of noisy advice. For the upper bound, we combine two ideas: our Pareto-optimal schedule for untrusted advice as well as a new algorithm for robust search in arrays, which can be seen as a problem ``dual'' to the one studied in~\cite{RivestMKWS80}. Here we show how results based on R\'enyi-Ulam types of games can be useful in the context of online computation with noisy advice. The resulting solution improves significantly in comparison to the upper bound of~\cite{DBLP:conf/aaai/0001K21}. In particular, we show that the performance of our schedule improves as function of $k$, whereas the schedule~\cite{DBLP:conf/aaai/0001K21} has performance independent of $k$.

We then turn our attention to lower bounds in the noisy advice setting. Here, we combine again two ideas. First, we abstract any schedule with noisy advice as a selection from a space ${\cal X}$ of candidate schedules, say of size $l$: given an ordering $X_0, \ldots, X_{l-1}$ 
of the schedules in ${\cal X}$ (unknown to the algorithm), in decreasing order of performance, any schedule with noisy advice will choose one of them, say $X_j$. We give a lower bound on this index $j$ using again ideas from robust search in arrays, under the assumption that the advice bits are responses to subset queries. Next, we need to establish a {\em measure} that relates the inefficiency of the chosen schedule of rank $j$ relative to the optimal one, i.e., the schedule $X_0$. We accomplish this by relating this problem to multi-processor contract scheduling with {\em fault tolerance} (without any advice).

In Section~\ref{sec:ft} we show how our techniques can be applied to problems outside the advice setting. In particular, we introduce and study a {\em robust} version of {\em fault-tolerant} 
contract scheduling. In the original problem~\cite{kupavskii2018lower}, the objective is to devise a contract schedule over $p$ multiple processors of optimal acceleration ratio, such that up to $f$ processor faults can be tolerated, for given $f<p$. In our problem, we additionally require that the schedule performs well even if all but a single processor are faulty. We combine ideas from our analysis in the untrusted and noisy settings, so as to obtain a schedule with an optimal tradeoff between these two objectives. 


\section{Preliminaries}
\label{sec:preliminaries}

As explained above, we adopt the abstraction of contract scheduling.
We give some useful definitions, following the notation of~\cite{DBLP:conf/aaai/0001K21}. 
A contract schedule $X$ is a sequence $(x_i)_{i=0}^\infty$, with $x_i \in \mathbb{R}^+$, and 
$x_{i+1} >x_i$, for all $i$. We call $x_i$ the {\em length} of {\em contract} $i$ in $X$. 
We will denote by $T$ the interruption time, and we make the standard assumption that an interruption can only occur after the first contract has completed its execution, since no schedule can have finite acceleration ratio otherwise. 

In the absence of advice, the worst-case acceleration ratio of $X$ is given by~\eqref{eq:bidding}. It is easy to see that the worst-case interruptions occur infinitesimally prior to the completion of a contract. In the presence of advice, the acceleration ratio of $X$ is equal to $T/\ell(X,T)$, as defined in Section~\ref{subsec:contract.scheduling}. The {\em consistency} of $X$ is the acceleration ratio of $X$ assuming no advice error, i.e., $\eta=0$, whereas its {\em robustness} is the acceleration ratio assuming adversarial advice. Note, in particular, that the robustness of a schedule is equal to its acceleration ratio without advice, and we will use these two terms interchangeably.  A schedule is called {\em Pareto-optimal}, if its consistency and robustness are in a Pareto-optimal relation. 
We call a schedule of robustness at most $r$, {\em $r$-robust}  (similarly for consistency).

We call a schedule {\em geometric with base $b>1$} if it is of the form $(b^i)_{i=0}^\infty$, and we denote it by $G_b$. Geometric schedules are important since they often lead to optimal solutions to many contract scheduling variants. It is known that the 
robustness of $G_b$ is equal to $b^2/(b-1)$~\cite{steins}.  This expression is minimized for $b=2$, thus $G_2$ has optimal acceleration ratio equal to 4. 
For any $r\geq 4$ define $\zeta_{1,r}$ and $\zeta_{2,r}$ as the smallest and largest roots of the function $\frac{x^2}{x-1}-r$. This definition implies the following property:

\begin{property}
 The schedule $G_b$ with $b>1$ has acceleration ratio at most $b^2/(b-1)$. In particular, for any $r \geq 4$, and 
 $b \in [\zeta_{1,r}, \zeta_{2,r}]$, $G_b$  has acceleration ratio at most $r$.
\label{prop:roots}
\end{property}
We have the following closed forms for the two roots:
\begin{equation}
\zeta_{1,r}=\frac{r-\sqrt{r^2-4r}}{2} \quad \textrm {and } \ \quad
\zeta_{2,r}=\frac{r+\sqrt{r^2-4r}}{2}
\end{equation}

All the above definitions apply to the single processor setting. In our work, we will show and exploit connections
between contract scheduling in a single processor with advice of size $k$, and contract scheduling without advice in $2^k$ processors. 
Hence, we present some definitions and notation concerning the setting of 
$p>1$ identical processors, labeled from the set $\{0, \ldots ,p-1\}$. In a contract schedule over this set of processors, each 
processor $j$ is assigned a sequence of contracts of the form   
$X_j=(x_{i,j})_{i=0}^\infty$, with $x_{i+1,j} > x_{i,j}$, for all $i$. We call the set 
$X=\{X_j\}_{j=0}^{p-1}$ a {\em $p$-processor} schedule, or equivalently, we say that $X$ is {\em defined by the set $\{X_j\}_{j=0}^{p-1}$}. The acceleration ratio of a $p$-processor schedule is defined as the worst-case ratio of the interruption time $T$, divided by the length of the longest contract completed by time $T$ among {\em all} processors. 

Let $X=(x_i)_{i=0}^\infty$ denote a sequence of positive numbers. We define $\alpha_X$ as the upper limit
\[
\alpha_X =\limsup_{n \rightarrow \infty} x_n^{1/n}.
\]
This quantity appears crucially in the statement of Gal's theorem. Last, given a set $X$ of positive reals, we define by $\bar{X}$ the sequence of all elements in $X$ in non-increasing order.

\section{Pareto-optimal contract scheduling with untrusted advice}
\label{sec:tradeoff}

In this section we show a tight lower bound on the consistency of any $r$-robust schedule with advice of size $k$, in the untrusted advice model.
We first give an overview of the proof. The starting observation is that any $r$-robust schedule $X$ with $k$ advice bits must be selected from a set ${\cal X}$ of at most $2^k$ $r$-robust schedules. Moreover, the consistency of $X$ is the acceleration ratio of the best schedule in ${\cal X}$, or, equivalently, the acceleration ratio of the $2^k$-processor schedule defined by ${\cal X}$. In Lemma~\ref{lemma:multi-search.alpha} we give a lower bound on this acceleration ratio as function of the parameter $\alpha_{\bar{\cal X}}$. Here, $\bar{\cal X}$ is defined as the sequence that consists of all contract lengths of schedules in ${\cal X}$, in non-decreasing order.
Next, we show that $\alpha_{\bar{\cal X}}$ must be within a certain range, namely within 
$[\zeta_{1,r}^{1/2^k}, \zeta_{2,r}^{1/2^k}]$ Corollary~\ref{cor:merge}). This is accomplished by first showing upper and lower bounds on the contract lengths of any $r$-robust schedule (Theorem~\ref{thm:zetas}). Combining the above yields the lower bound. The tightness of the result will follow by directly comparing to the schedule (upper bound) of~\cite{DBLP:conf/innovations/000121}.

We proceed with the technical results. The following lemma is a special case of a more result that we prove 
in Section~\ref{subsec:noisy.lower}, namely Theorem~\ref{thm:multi-search.alpha.faulty}.

\begin{lemma}
Let $X$ be a $p$-processor schedule, as defined by a set ${\cal X}$ of $2^k$ schedules, of finite
acceleration ratio. Then
\[
\acc(X) \geq \frac{\alpha_{\bar{\cal X}}^{2^k+1}}{\alpha_{\bar{\cal X}}^{2^k}-1}.
\]
\label{lemma:multi-search.alpha}
\end{lemma}

In the next step, we show that $\zeta_{2,r}^i$ and $\zeta_{1,r}^i$ are (roughly) 
upper and lower bounds on the lengths of any $r$-robust schedule. 
\begin{theorem}[Appendix]
Let $X=(x_i)_{i=0}^\infty$  be an $r$-robust schedule, for some fixed, finite $r\geq 4$. Then
\[
x_i\leq   c\cdot \zeta_{2,r}^i \ \textrm{and } x_i \geq  d \cdot \zeta_{1,r}^i,
\]
if $r>4$. Moreover, 
\[
x_i\leq   c\cdot i \cdot 2^i \ \textrm{and } x_i \geq  d  \cdot 2^i,
\]
if $r=4$, where $c$,$d$ are only functions of $r$. 
\label{thm:zetas}
\end{theorem}

Note that the bounds of Theorem~\ref{thm:zetas} are tight. From Property~\ref{prop:roots}, any geometric schedule 
$G_b$ with base $b\in  [\zeta_{1,\rho}, \zeta_{2,\rho}]$ is $r$-robust.
We also obtain the following useful corollary, which follows from Theorem~\ref{thm:zetas} and the definition of $\alpha_{\bar X}$.

\begin{corollary}[Appendix]
Let $X$ be a $p$-processor schedule defined by the set ${\cal X}=\{X_0,X_1, \ldots ,X_{p-1}\}$, where each $X_j$ is an $r$-robust schedule, 
for a given $r\geq 4$. Then 
\[
\alpha_{\bar {\cal X}} \in [\zeta_{1, r}^{1/p}, \zeta_{2, r}^{1/p}].
\] 
\label{cor:merge}
\end{corollary}

We now give the statement and the proof of our lower bound. 

\begin{theorem}
For all $r\geq 4$, any $r$-robust schedule with untrusted advice of size $k$ has consistency at least
\begin{equation}
\min \frac{x^{2^{k}+1}}{x^{2^k}-1} \quad \textrm{subject to} \quad  \zeta_{1,r}^{1/2^k} \leq x \leq \zeta_{2,r}^{1/2^k}.
\label{eq:pareto.optimal}
\end{equation} 
\label{thm:lower.pareto}
\end{theorem}

\begin{proof}
Any schedule $X$ with advice of size $k$ will choose a schedule among a set of at most $2^k$ schedules, say ${\cal X}=\{X_0, \ldots X_{2^k-1}\}$. For $X$ to be $r$-robust, it must be that each $X_i$, with $i \in [0, 2^k-1]$ is likewise $r$-robust, otherwise maliciously generated advice would chose a schedule of robustness greater than $r$. 

Note that the consistency of $X$ is identical to the acceleration ratio of the $2^k$-processor schedule that is defined by ${\cal X}$. Let $X'$ denote this multi-processor schedule. From Lemma~\ref{lemma:multi-search.alpha}, we have that $\acc(X') \geq \frac{\alpha_{\bar{\cal X}}^{2^k+1}}{\alpha_{\bar{\cal X}}^{2^k}-1}$. Last, since every single-processor schedule $X_i$ must be $r$-robust, from Corollary~\ref{cor:merge} it follows that
$\alpha_{\bar{\cal X}} \in [ \zeta_{1,r}^{1/2^k}, \zeta_{2,r}^{1/2^k} ]$. 
\end{proof}

We now argue that our lower bound is tight. Consider the following set of schedules (see Figure~\ref{fig:cyclic} for an illustration).

\begin{definition}
For given $b>1$, and $l \in \mathbb{N}^+$ define ${\cal X}_{b,l}$ as the set of schedules
$\{X_0, \ldots X_{l-1}\}$, in which $X_i=(b^{i+jl})_{j=0}^\infty$, for all $i\in [0,l-1]$.
\label{def:bunch.old}
\end{definition}

\begin{figure}
	\centering
	\includegraphics[scale=.3]{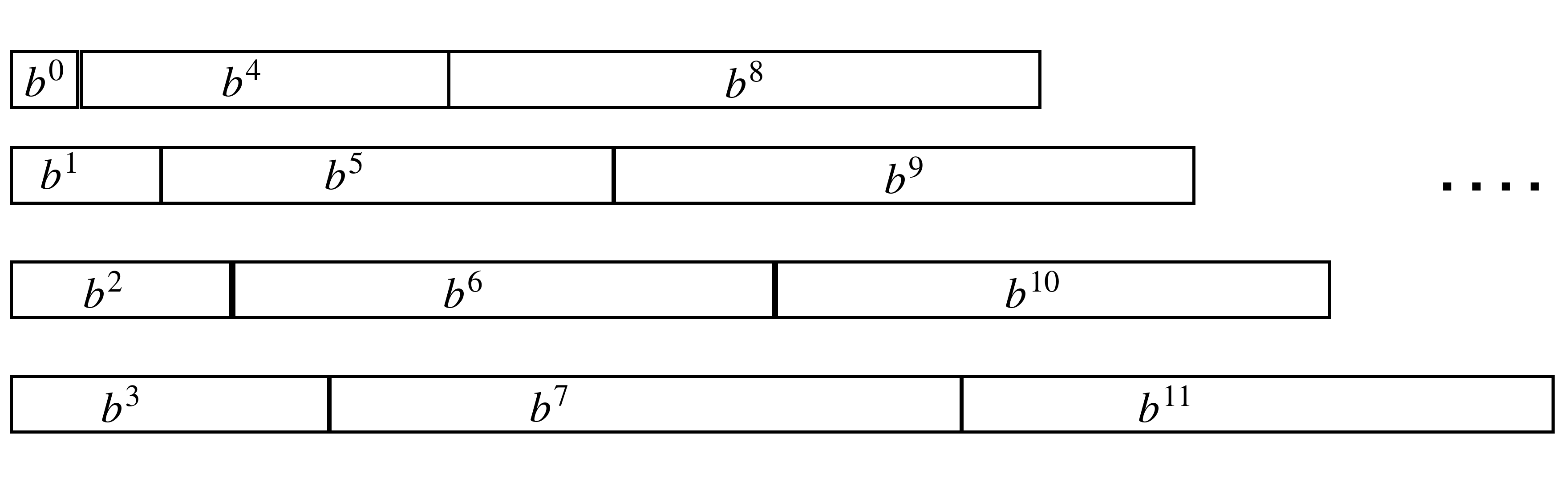}
	\caption{Illustration of the schedule ${\cal X}_{b,4}$.}
	\label{fig:cyclic}
\end{figure}

The following upper bound is from~\cite{DBLP:conf/aaai/0001K21}, which we restate in a way that will be useful in our study of the noisy model, in the next section.

\begin{theorem}
We can find $b>1$ such that: i) each schedule in ${\cal X}_{b,2^k}$ is $r$-robust; and ii) the best schedule in ${\cal X}_{b,2^k}$ has consistency that matches the lower bound of Theorem~\ref{thm:lower.pareto}; i.e., this schedule is Pareto-optimal for untrusted advice of size $k$. 
\label{thm:pareto.upper}
\end{theorem}

\begin{proof}
Each schedule $X_i$ in ${\cal X}_{b,2^k}$ is essentially geometric with base 
$b^{2^k}$ (with a multiplicative offset equal to $b^i$), hence, from Property~\ref{prop:roots}, $X_i$ has acceleration ratio at most $b^{2^{k+1}}/(b^{2^k}-1)$. From the same property, as long as $\zeta_{1,r} \leq b^{2^k} \leq \zeta_{2,r}$, each such schedule $X_i$ is $r$-robust. In addition, following the proof of Theorem 3 in~\cite{DBLP:conf/aaai/0001K21}, 
the best schedule in ${\cal X}_{b,2^k} $ has acceleration ratio at most
$b^{2^k+1}/(b^{2^k}-1)$, which is, therefore, the consistency of the schedule chosen by the advice. Selecting $b \in [\zeta_{1,r}^{1/2^k}, \zeta_{2,r}^{1/2^k}]$, so that $b^{2^k+1}/(b^{2^k}-1)$ is minimized, yields an $r$-robust schedule whose consistency matches the one of Theorem~\ref{thm:lower.pareto}, and is therefore optimal. Note that finding the appropriate $b>1$ is computationally efficient, since it only involves
finding the roots $\zeta_{1,r}$ and $\zeta_{2,r}$.
\end{proof}

We conclude with an observation concerning the class of schedules ${\cal X}_{b,l}$ that we will use in the following section. 
From the definition of ${\cal X}_{b,l}$ (and in particular, the cyclic assignment of contract lengths over the different schedules), for any interruption time $T$, there is a cyclic permutation 
$\pi$ of $\{0, \ldots l-1\}$ which determines the ordering of the schedules in ${\cal X}_{b,l}$ in terms of their performance. In other words, $X_{\pi(0)}$ is the best schedule that completes the 
longest contract by time $T$, and $X_{\pi(2^k-1)}$ is the worst schedule that completes the smallest contract by time $T$. We will refer to this property as the {\em cyclic property of ${\cal X}_{b,l}$}.

\section{Contract scheduling with noisy advice}
\label{sec:noisy}

In this section, we study contract scheduling under the noisy advice model. Namely, the scheduler has access to an advice string $\sigma$ of size $k$, with at most $H$ erroneous bits, 
for some known $H \in [0,k]$. 
We will give both upper and lower bounds on the performance of schedules, in Sections~\ref{subsec:noisy.upper} and~\ref{subsec:noisy.lower}, respectively, which apply to the standard noisy model. In 
Section~\ref{subsec:noisy.discussion} we discuss these bounds, the improvements we obtain, and explain how they can be extended to the robust noisy model.

\subsection{Upper bound}
\label{subsec:noisy.upper}

The idea behind the upper bound is as follows. We would like to use the advice so as to select a good schedule from the 
set ${\cal X}_{b,2^k}$ (see Definition~\ref{def:bunch.old}), from some suitable $b>1$. From the cyclic property of ${\cal X}_{b,2^k}$, we know
that  there is a cyclic ordering of its schedules in decreasing order of performance, but the algorithm does not know the precise ordering. To model this situation, we define the following problem which we call {\sc MinCyclic}.  
The input to this problem is an array $A[0\ldots n-1]$ whose elements are an unknown cyclic permutation of $\{0, \ldots ,n-1\}$. 
The objective is to use $k$ queries, at most $H$ of which can be erroneous, so as to identify an index of the array whose element is as small as possible. 

Following the notation of~\cite{RivestMKWS80}, we define the partial sum of binomial coefficients as

\[
  \mch{N}{m}   := \sum_{j=0}^m {N \choose j}, \  \textrm{for $m \leq N$}.
\]

We can use an algorithm of~\cite{RivestMKWS80} so as to obtain an algorithm for {\sc MinCyclic}. Note that in~\cite{RivestMKWS80}, the objective is to minimize the number of comparison queries while searching with error. For our setting, a comparison query\footnote{As explained in~\cite{DBLP:conf/aaai/0001K21}, both comparison and subset queries have an equivalent, and natural interpretation of the form ``Does the interruption $T$ occur in a certain subdivision of the timeline''? This makes such queries useful in a more practical context.} asks a question of the form: ``For the unknown 
index $x \in [0, \ldots ,n-1]$, for which $A[x]=0$, is $x \leq B$, for some chosen $B$''?
\begin{theorem}[Appendix]
There is an algorithm for {\sc MinCyclic} based on comparison queries that 
outputs an index $j$ such that $A[j] \leq \lceil n\mch{k-H}{H}/2^{k-H} \rceil$, for all $H \leq k/2$.
\label{thm:cyclic.upper}
\end{theorem}

Combining Theorem~\ref{thm:cyclic.upper} with the cyclic property of ${\cal X}_{b,2^k}$ we have the following result. 

\begin{theorem}
In the noisy advice setting, there is a schedule based on comparison queries of acceleration ratio at most
\[
\frac{1+U}{2^k}\left(1+\frac{2^k}{1+U}\right) ^{1+\frac{1+U}{2^k}}, \quad 
\textrm{where $U=2^H \mch{k-H}{H}$}.
\]
\label{thm:noisy.upper}
\end{theorem}

\begin{proof}
We apply the algorithm of Theorem~\ref{thm:cyclic.upper} on the set of indices of all 
schedules in ${\cal X}_{b,2^k}$ with $n=2^k$, where $b>1$ will be chosen later. The output is the index of a schedule in ${\cal X}_{b,2^k}$ which is ranked at most $U$ among the schedules in ${\cal X}_{b,2^k}$
(where lower ranking indicates better contract length completed by the interruption time). From the definition of ${\cal X}_{b,2^k}$, and in particular its cyclic property, this means that the selected schedule completes a contract of length at most $b^U$ times larger than 
the best schedule in ${\cal X}_{b,2^k}$. Following the proof of Theorem~\ref{thm:pareto.upper}, we infer that the acceleration ratio of the selected schedule is at most $\frac{b^{2^k+1+U}}{b^{2^k}-1}$.  This expression is minimized for 
$b=(\frac{2^k+U+1}{U+1})^{1/2^k}$, from which we obtain that the acceleration ratio of our schedule is at most
\[
\frac{1+U}{2^k}\left(1+\frac{2^k}{1+U}\right) ^{1+\frac{1+U}{2^k}}.
\]
\end{proof}

\subsection{Lower bound}
\label{subsec:noisy.lower}

The idea behind the lower bound is as follows. With $k$ advice bits, the best one can do is 
chose the best schedule from a set ${\cal X}$ that consists of at most $2^k$ schedules. If the advice were error-free, $|{\cal X}|$ could be as large as $2^k$; however, in the presence of errors, the algorithm may chose to narrow $|{\cal X}|$.

Our approach follows the combination of two ideas. The first idea uses an abstraction for searching an element within a space $A$ of $n$ elements indexed in $[0, \ldots, n-1]$, and which are 
ranked according to an unknown ranking function 
$\pi$. Namely, the element with index $\pi(0)$ is the ``best'' element in $A$, the one with index $\pi(1)$ is the second best, and the one with index 
$\pi(n-1)$ is the worst. The algorithm is allowed to ask $k$ queries, $H$ of which may be erroneous, for some known $H\leq k$. The objective is to find an element $j$ in the set $A$, such that $\pi(j)$ is minimized.
 We denote this problem as {\sc Search}. The following result establishes a lower bound on the rank that can be achieved by an algorithm that uses subset queries. This will allow us, in turn, to place a lower-bound on the rank of the chosen schedule within ${\cal X}$. Here, a subset query is of the form ``Does the element of rank $\pi(0)$ belong in a certain subset of $\{0, \ldots, n-1\}$''?

\begin{theorem}[Appendix]
For given $n,k,H$, the response oracle can force an algorithm for {\sc Search} to output 
en element $e$ of $A$ such that $\pi(e) \geq \lfloor n\mch{k}{H}/2^{k} \rfloor$. 
\label{thm:cyclic.lower}
\end{theorem}

The second idea is to define a {\em measure} that relates how much worse a schedule of rank $j$ in ${\cal X}$ has to be relative to the best schedule in ${\cal X}$. We will accomplish this by appealing to the concept of {\em fault tolerance}. More precisely, given integers $p$, $\phi$, with $\phi<p$, we define $\mathtt{FCS}(p,\phi)$ as the problem of finding a $p$-processor schedule of minimum acceleration ratio, if up to $\phi$ processors may be faulty. Here, a processor failure means that the processor may cease to function as early as $t=0$, and thus may not contribute any completed contracts. To solve this problem, we will rely on Gal's {\em functional theorem}, stated below. 

\begin{theorem}[Gal~\cite{gal:general}]
Let $q$ be a positive integer, and $X=(x_i)_{i=0}^\infty$ a sequence of positive numbers with 
$\sup_{n \geq 0} x_{n+1}/x_n <\infty$ and $\alpha_X>0$. Suppose that $F_i$ is a sequence of functionals that satisfy the following properties:
\begin{itemize}
\item[(1)]$ F_i(X)$ depends only on $x_0,x_1, \ldots x_{i+q}$,
\item[(2)] $F_i(X)$ is continuous in every variable, for all positive sequences $X$,
\item[(3)] $F_i(a X)=F_i(X)$, for all $a>0$,
\item[(4)] $F_i(X+Y) \leq \max(F_i(X), F_i(Y))$, for all positive sequences $X,Y$, and
\item[(5)] $F_{i+j}(X) \geq F_i(X^{+j})$, for all $j \geq 1$, where $X^{+j}=(x_j, x_{j+1}, \ldots)$.  
\end{itemize}
Then 
\[
\sup_{0 \leq k <\infty} F_k(X) \geq \sup_{0 \leq k<\infty} F_k (G_{\alpha_X}), 
\]
where $G_a$ is defined as the geometric sequence $(a^i)_{i=0}^\infty$.
\label{thm:gal}
\end{theorem}

The next theorem is the main technical result for $\mathtt{FCS}(p,\phi)$, that allows us to apply Gal's theorem. 
The proof extends ideas of~\cite{aaai06:contracts} to the fault-tolerant setting.  To keep the notation consistent, we will denote by $\bar{X}$ the sequence of all contract lengths, in non-decreasing order, in the $p$-processor schedule $X$. 

\begin{theorem}[Appendix]
Every schedule $X$ for $\mathtt{FCS}(p,\phi)$ that has finite acceleration ratio satisfies
\[
\acc(X) \geq \frac{\alpha_{\bar{X}}^{p+1+\phi}}{\alpha_{\bar{X}}^p-1}.
\]
\label{thm:multi-search.alpha.faulty}
\end{theorem}

We now show how to obtain the lower bound, by combining the above ideas.

\begin{theorem}
For every schedule $X$ and $k$ subset queries in the noisy advice model, 
we have $\acc(X) \geq \frac{1}{L}(1+L)^{1+1/L}$, 
where $L=2^k/\mch{k}{H}$.
\label{thm:noisy.lower}
\end{theorem}

\begin{proof}
Every algorithm for the problem  will use the $k$ advice bits so as to select a schedule from a set ${\cal X}=\{X_0, \ldots , X_{l-1}\}$ of candidate schedules,  for some $l \leq 2^k$. For a given interruption time, there is an ordering of the $l$ schedules in 
${\cal X}$ such that, $X_{\pi(i)}$ has no worse acceleration ratio than $X_{\pi(i+1)}$, namely the permutation orders the schedules in decreasing order of their performance. From Theorem~\ref{thm:cyclic.lower}, it follows that the algorithm will chose a schedule 
$X_{j}$ such that $\pi(j)\geq \lfloor l\mch{k}{H}/2^{k} \rfloor$. The acceleration ratio of the selected schedule is at least the acceleration ratio of the $l$-processor schedule defined by 
${\cal X}$, 
 in which up to $\phi_l = \lfloor l\mch{k}{H}/2^{k} \rfloor$ processors may be faulty. From Theorem~\ref{thm:multi-search.alpha.faulty},
\begin{equation}
\acc(X) \geq \frac{\alpha_{\bar{X}}^{l+1+\phi_l}}{\alpha_{\bar{X}}^l-1}, \quad \textrm{with $\phi_l = \lfloor l\mch{k}{H}/2^{k} \rfloor$}.
\label{eq:nasty}
\end{equation}
We now consider two cases. Suppose first that $l<L$. In this case,
case $\phi_l=0$, and therefore~\eqref{eq:nasty} implies that $\acc(X) \geq \alpha_{\bar{X}}^{l+1}/(\alpha_{\bar{X}}^l-1)$, which is minimized for $\alpha_{\bar{X}}=(l+1)^{1/l}>1$, therefore $\acc(X) \geq \frac{1}{l}(l+1)^{1+1/l}$. This function is decreasing in $l$, and since $l<L$ we have
\[
\acc(X) \geq \frac{1}{L}(1+L)^{1+1/L}.
\]
Next, suppose that $l \in [L,2^k]$. In this case,~\eqref{eq:nasty} gives 
\[
\acc(X) \geq \frac{\alpha_{\bar{X}}^{l(1+1/L)}}{\alpha_{\bar{X}}^l-1}.
\] 
The above expression is minimized for $\alpha_{\bar{X}} =(1+L)^{1/l}$, and by substitution we obtain again
\[
\acc(X) \geq \frac{1}{L}(1+L)^{1+1/L}.
\]
\end{proof}

\subsection{Extensions and comparison of our bounds}
\label{subsec:noisy.discussion}

\subsubsection{Extension to the robust model}

The upper and lower bounds of Sections~\ref{subsec:noisy.upper} and~\ref{subsec:noisy.lower} apply to the standard noisy model. However, our techniques are readily applicable to the robust noisy model.  For the upper bound, we need to add a constraint that enforces the requirement that each schedule in ${\cal X}_{b,2^k}$ is $r$-robust. From the proof of Theorem~\ref{thm:pareto.upper}, we know that this constraint can be expressed as $\zeta_{1,r} \leq b^{2^k} \leq \zeta_{2,r}$. Therefore, we obtain a lower bound by finding $b>1$ that minimizes the expression 
\[
\min \frac{b^{2^k+1+U}}{b^{2^k}-1} \quad \textrm{subject to $\zeta_{1,r} \leq b^{2^k} \leq \zeta_{2,r}$}.
\] 

In a similar vein, for the lower bound, we can follow the proof of Theorem~\ref{thm:noisy.lower} but also require that each of the $l$ schedules in ${\cal X}$ is $r$-robust. 
From Corollary~\ref{cor:merge}, using the notation of Section~\ref{subsec:noisy.lower}, this means the additional constraint 
$ \zeta_{1,r} \leq \alpha_{\bar{X}}^{l} \leq \zeta_{2,r}$.
We can thus optimize the functions discussed in the proof, subject to the above constraint.

Note the similarity between the constraints concerning the upper and lower bounds. This implies that small gaps in the standard noisy
model will intuitively remain small in the robust noisy model.

\subsubsection{Improvements and comparison of our bounds}

We compare our upper bound to the upper bound of~\cite{DBLP:conf/aaai/0001K21}, and for simplicity we focus on the standard noisy model\footnote{The schedule of~\cite{DBLP:conf/aaai/0001K21} applies to the robust noisy model. However, since it is based on a collection of 
schedules of the form ${\cal X}_{b,k}$, we can translate their result to the standard noisy model. Note that in~\cite{DBLP:conf/aaai/0001K21} the size of ${\cal X}$ is only $k$.}. As the number of advice bits $k$ becomes
large, the upper bound of Theorem~\ref{thm:noisy.upper} is very close to $f(2^k/U)$, where $f(x)=\frac{1}{x}(1+x)^{1+1/x}$. Although $U$ does not have a closed form, we can use the following useful approximation for the partial sum of binomial coefficients~\cite{macwilliams1977theory}. Denote by ${\cal H}$ the binary entropy function. Then

\begin{equation}
\frac{2^{N{\cal H}(\frac{m}{N})}}{\sqrt{8m(1-\frac{m}{N})}}
\leq \mch{N}{m} \leq 2^{N{\cal H}(\frac{m}{N})}, \quad \textrm{for $0<m<N/2$.} 
\label{eq:amazing}
\end{equation}

Let $\tau=H/k$ denote the upper bound on the fraction of the erroneous bits in the advice string, and suppose that $\tau \leq 1/2$. 
Then~\eqref{eq:amazing} implies that 
\[
\frac{2^k}{U}\geq \frac{2^k}{2^H 2^{(k-H){\cal H}(\frac{\tau}{1-\tau})}}=
2^{k(1-\tau)(1-{\cal H}(\frac{\tau}{1-\tau}))}.
\]
In contrast, the schedule of~\cite{DBLP:conf/aaai/0001K21} has acceleration ratio $f(2H/k)$, 
and note that, using the above definitions, $2H/k=2\tau$. 

Since $f(x)$ is decreasing in $x$, and tends to $1$ as $x\rightarrow \infty$, we observe that for
$\tau \leq 1/2$,  the acceleration ratio of our schedule approaches 1 as $k \rightarrow \infty$. In contrast, the schedule
of~\cite{DBLP:conf/aaai/0001K21} has acceleration ratio roughly $f(2H/k)=f(2\tau)$,
which is constant and independent of $k$, again for $\tau \leq 1/2$. Hence, we obtain substantially better performance.

Next, we compare our upper and lower bounds against each other. The lower bound of Theorem~\ref{thm:noisy.lower} equals $f(L)$, 
where recall that $L=2^k/\mch{k}{H}$. For large $x$, $f(x)\approx 1+ 1/x$. Thus, for large $k$, we can approximate the lower bound 
as $1+1/L$. A rather crude approximation of the upper bound is $1+1/(2^H L)$, but a better, albeit more complex approximations can be obtained by using the left-hand side inequality of~\eqref{eq:amazing}.
Last, note that $L$ increases exponentially with $k$, since $L \geq 2^{k(1-{\cal H}(\tau))}$.

\section{An application: Robust fault-tolerant contract scheduling}
\label{sec:ft}

The techniques we introduced can find applications in multi-criteria optimization problems outside the realm of computation with
advice. We illustrate such an example, in the context of contract scheduling. We define the {\em robust fault-tolerant contract scheduling} problem as follows.
Given $r \in \mathbb{R}$ with $r\geq 4$, $p \in \mathbb{Z}$, 
and $f \in \mathbb{Z}$, with $f<p$, find a $p$-processor schedule $X$ which has minimum acceleration ratio,
if up to $f$ processors are faulty, but also has acceleration ratio at most $r$, if all but a single processor are faulty. 
This is an extension of the fault-tolerant model of~\cite{kupavskii2018lower}, in which we treat $f$ as a ``soft'' bound on the number of faults that 
may occur, and would still like the schedule to perform well if this bound is exceeded. We denote this problem as 
$\mathtt{RFT}(r,p,f)$.

\begin{theorem}[Appendix]
There is an optimal schedule for $\mathtt{RFT}(r,p,f)$ of acceleration ratio
\[
\min  \frac{b^{p+f+1}}{b^p-1} \quad \textrm{subject to} \quad b \in [\zeta_{1,r}^{1/p}, \zeta_{2,r}^{1/p}].
\]
\label{thm:rft}
\end{theorem}

\bibliographystyle{plain}
\bibliography{fault-contract,targets}

\newpage

\bigskip
\appendix
{\Large \bf Appendix}

\section{Omitted proofs of Section~\ref{sec:tradeoff}}

We will first give the proof of Theorem~\ref{thm:zetas}, and in fact prove a somewhat stronger statement, in which we will give explicit formulas for the constants $c$ and $d$.

\begin{proposition}\label{recurrence:1}
	Let $(x_i)_{i=0}^\infty$ and $(y_i)_{i=0}^\infty$ be such that
	\begin{equation}\label{bound:0}
		\sum_{i=0}^nx_i\leq r x_{n-1}
		\quad\text{and}\quad
		\sum_{i=0}^ny_i= r y_{n-1},
	\end{equation}
	for all $n\geq 2$, where $r\geq 4$ is a fixed constant. If $x_0\geq y_0$ and $x_1\leq y_1$, then $x_i\leq y_i$, for all $i\geq 2$.
\end{proposition}

\begin{proof}
	First, we introduce the sequence of coefficients $(A_k,B_k)_{k=1}^\infty$ defined recursively by
	\begin{equation*}
		\begin{pmatrix}
			A_{k+1}\\B_{k+1}
		\end{pmatrix}
		=M
		\begin{pmatrix}
			A_{k}\\B_{k}
		\end{pmatrix},
	\end{equation*}
	where
	\begin{equation*}
		M=\begin{pmatrix}
			r-1&-1\\1&1
		\end{pmatrix},
	\end{equation*}
	with initial values $(A_1,B_1)=(r-1,1)$. We claim that
	\begin{equation}\label{bound:1}
		A_k,B_k\geq 0,\quad\text{for every }k\geq 1.
	\end{equation}
	Let us momentarily assume the validity of \eqref{bound:1} and complete the proof of the proposition.
	
	To that end, we show that
	\begin{equation}\label{bound:2}
		x_{n}\leq A_kx_{n-k}-B_k\sum_{i=0}^{n-k-1}x_i
		\quad\text{and}\quad
		y_{n}= A_ky_{n-k}-B_k\sum_{i=0}^{n-k-1}y_i,
	\end{equation}
	for all $1\leq k<n$. This follows from an induction argument on $k$. Indeed, for a given $n\geq 2$, the case $k=1$ is a mere reformulation of \eqref{bound:0}. Then, in view of \eqref{bound:1}, assuming that \eqref{bound:2} holds for some $1\leq k\leq n-2$, we obtain
	\begin{equation*}
		\begin{aligned}
			x_{n}&\leq A_kx_{n-k}-B_k\sum_{i=0}^{n-k-1}x_i
			\leq A_k(rx_{n-k-1}-\sum_{i=0}^{n-k-1}x_i)-B_k\sum_{i=0}^{n-k-1}x_i
			\\
			&= (A_k(\rho-1)-B_k)x_{n-k-1}-(A_k+B_k)\sum_{i=0}^{n-k-2}x_i
			= A_{k+1}x_{n-k-1}-B_{k+1}\sum_{i=0}^{n-k-2}x_i
		\end{aligned}
	\end{equation*}
	and similarly for $y_n$,
	thereby establishing \eqref{bound:2} for all $1\leq k<n$.
	
	Now, observe that setting $k=n-1$ in \eqref{bound:2} yields
	\begin{equation}\label{bound:3}
		x_{n}\leq A_{n-1}x_{1}-B_{n-1}x_0
		\leq A_{n-1}y_{1}-B_{n-1}y_0=y_n,
	\end{equation}
	for all $n\geq 2$, which completes the proof.
	
	Thus, there only remains to justify the bound \eqref{bound:1}, which will easily follow from an explicit representation formula for $(A_k,B_k)$ based on an eigenvector decomposition of $M$. More precisely, straightforward calculations establish that the eigenvalues of $M$ are the two roots $\zeta_{1,r}\leq \zeta_{2,r}$ of the characteristic polynomial $p(\zeta)=\zeta^2-r\zeta+r$, which are given explicitly in Section~\ref{sec:preliminaries}.
	They are both positive if $r\geq 4$ and distinct whenever $r>4$. In fact, it is readily seen that $\zeta_{2,r}\geq \zeta_{1,r}>1$. Moreover, it holds that $\zeta_{2,r}+\zeta_{1,r}=r$, $\zeta_{2,r}\zeta_{1,r}=r$ and $(\zeta_{2,r}-1)(\zeta_{1,r}-1)=1$.
	
	When $r>4$, we obtain the eigenvector decomposition
	\begin{equation}\label{bound:4}
		\begin{pmatrix}
			A_k\\B_k
		\end{pmatrix}
		=
		\frac{\zeta_{2,r}^k}{\zeta_{2,r}-\zeta_{1,r}}
		\begin{pmatrix}
			\zeta_{2,r}-1\\1
		\end{pmatrix}
		-
		\frac{\zeta_{1,r}^k}{\zeta_{2,r}-\zeta_{1,r}}
		\begin{pmatrix}
			\zeta_{1,r}-1\\1
		\end{pmatrix}
	\end{equation}
	for every $k\geq 1$, which implies \eqref{bound:1} because $\zeta_{2,r}\geq \zeta_{1,r}>1$. Finally, further letting $r\to 4$ yields the representation
	\begin{equation}\label{bound:5}
		\begin{pmatrix}
			A_k\\B_k
		\end{pmatrix}
		=
		\begin{pmatrix}
			2^k+k2^{k-1}\\ k2^{k-1}
		\end{pmatrix}
	\end{equation}
	in the case $r=4$, which also validates \eqref{bound:1} and thus concludes the proof of the proposition.
\end{proof}

\begin{remark}
	Observe from \eqref{bound:3}, \eqref{bound:4} and \eqref{bound:5} in the preceding proof that one has the convenient representation formulas, for every $n\geq 2$,
	\begin{equation}\label{representation:1}
		\begin{aligned}
			y_n&=\frac{\zeta_{2,r}^{n-1}(\zeta_{2,r}-1)-\zeta_{1,r}^{n-1}(\zeta_{1,r}-1)}{\zeta_{2,r}-\zeta_{1,r}}y_1
			-\frac{\zeta_{2,r}^{n-1}-\zeta_{1,r}^{n-1}}{\zeta_{2,r}-\zeta_{1,r}}y_0
			\\
			&=\zeta_{2,r}^{n-1}\frac{(\zeta_{2,r}-1)y_1-y_0}{\zeta_{2,r}-\zeta_{1,r}}
			-\zeta_{1,r}^{n-1}\frac{(\zeta_{1,r}-1)y_1-y_0}{\zeta_{2,r}-\zeta_{1,r}},
		\end{aligned}
	\end{equation}
	if $r>4$, and
	\begin{equation}\label{representation:2}
		\begin{aligned}
			y_n&=(2^{n-1}+(n-1)2^{n-2})y_1-(n-1)2^{n-2}y_0
			\\
			&=2^{n-1}y_1+(n-1)2^{n-2}(y_1-y_0),
		\end{aligned}
	\end{equation}
	when $r=4$. If one further requires that \eqref{bound:0} hold for $n=1$, whereby $y_1=(r-1)y_0$, then one finds that
	\begin{equation*}
		y_n=\frac{\zeta_{2,r}^{n}(\zeta_{2,r}-1)-\zeta_{1,r}^{n}(\zeta_{1,r}-1)}{\zeta_{2,r}-\zeta_{1,r}}y_0,
	\end{equation*}
	if $r>4$, and
	\begin{equation*}
		y_n=\left(3\cdot2^{n-1}+(n-1)2^{n-1}\right)y_0,
	\end{equation*}
	when $r=4$.
\end{remark}

\begin{proposition}
	Let $(z_i)_{i=0}^\infty$ be a nonnegative sequence such that
	\begin{equation*}
		\sum_{i=0}^nz_i\leq r z_{n-1},
	\end{equation*}
	for all $n\geq 2$, where $r\geq 4$ is a fixed constant. Then, for every $i\geq 0$, one has the bounds
	\begin{equation*}
		(\zeta_{1,r}-1)\zeta_{1,r}^i z_0\leq z_{i+1}\leq \frac{\zeta_{2,r}^{i+1}-\zeta_{1,r}^{i+1}}{\zeta_{2,r}-\zeta_{1,r}}z_1,
	\end{equation*}
	where $\zeta_{2,r}$ and $\zeta_{1,r}$ are defined in Section~\ref{sec:preliminaries}. Furthermore, these bounds reduce to
	\begin{equation*}
		2^i z_0\leq z_{i+1}\leq (i+1)2^iz_1,
	\end{equation*}
	when $r=4$.
\end{proposition}

\begin{proof}
	The upper bounds follow directly from Proposition \ref{recurrence:1}, with $x_0=y_0=z_0$ and $x_1=y_1=z_1$, and the representation formulas \eqref{representation:1} and \eqref{representation:2}.
	
	The lower bounds are more subtle. In order to establish their validity, we define, for any given integer $j\geq 0$, an auxiliary sequence $(x_i)_{i=0}^\infty$ by
	\begin{equation*}
		x_0=\sum_{k=0}^jz_k
		\quad\text{and}\quad
		x_i=z_{j+i},\text{ if }i\geq 1.
	\end{equation*}
	In particular, it holds that
	\begin{equation*}
		\sum_{i=0}^nx_i\leq r x_{n-1},
	\end{equation*}
	for all $n\geq 2$.
	
	Therefore, by Proposition \ref{recurrence:1} combined with the representation formulas \eqref{representation:1} and \eqref{representation:2}, we deduce that, for all $n\geq 2$,
	\begin{equation*}
		x_n\leq \zeta_{2,r}^{n-1}\frac{(\zeta_{2,r}-1)x_1-x_0}{\zeta_{2,r}-\zeta_{1,r}}
		-\zeta_{1,r}^{n-1}\frac{(\zeta_{1,r}-1)x_1-x_0}{\zeta_{2,r}-\zeta_{1,r}},
	\end{equation*}
	if $r>4$, and
	\begin{equation*}
		x_n\leq 2^{n-1}x_1+(n-1)2^{n-2}(x_1-x_0),
	\end{equation*}
	when $r=4$. Since $(x_i)_{i=0}^\infty$ is also nonnegative and, as $n\to\infty$, the dominant terms above are $\zeta_{2,r}^{n-1}$ and $(n-1)2^{n-2}$, we conclude that necessarily $x_0\leq (\zeta_{2,r}-1)x_1$, for all values $r\geq 4$. In terms of the original sequence $(z_i)_{i=0}^\infty$, recalling that $(\zeta_{2,r}-1)(\zeta_{1,r}-1)=1$, this yields that
	\begin{equation*}
		(\zeta_{1,r}-1)\sum_{k=0}^jz_k\leq z_{j+1},
	\end{equation*}
	for every $j\geq 0$. In particular, if
	\begin{equation*}
		z_{i+1}\geq (\zeta_{1,r}-1)\zeta_{1,r}^i z_0
	\end{equation*}
	holds for every $0\leq i\leq j$, then one finds that
	\begin{equation*}
		\begin{aligned}
			z_{j+2}&\geq (\zeta_{1,r}-1)\sum_{k=0}^{j+1}z_k\geq (\zeta_{1,r}-1)z_0+(\zeta_{1,r}-1)^2z_0\sum_{k=1}^{j+1}\zeta_{1,r}^{k-1}
			\\
			& =(\zeta_{1,r}-1)z_0+(\zeta_{1,r}-1)^2z_0\frac{\zeta_{1,r}^{j+1}-1}{\zeta_{1,r}-1}= (\zeta_{1,r}-1)\zeta_{1,r}^{j+1} z_0,
		\end{aligned}
	\end{equation*}
	thereby completing the proof of the lower bounds, by induction. We thus have arrived at the proof of Theorem~\ref{thm:zetas}.
\end{proof}

Next, we give a proof of Corollary~\ref{cor:merge}. Again, we will prove a somewhat stronger statement, as expressed in the following lemma. The proof of Corollary~\ref{cor:merge} follows by replacing $c_A, c_B$ with $c,d$, and $A$, $B$, with $\zeta_{1,r}$, $\zeta_{2,r}$, respectively, as well as by setting $R=1$.

\begin{lemma}
	Let $(x_{1,i})_{i=1}^\infty$, $(x_{2,i})_{i=1}^\infty$, \ldots, $(x_{N,i})_{i=1}^\infty$ be $N$ positive nondecreasing sequences satisfying the bounds
	\begin{equation*}
		c_AA^i\leq x_{j,i}\leq c_Bi^RB^i,
	\end{equation*}
	for all $i$ and $j$, where $R\geq 0$, $c_A, c_B>0$ and $A,B>1$ are given constants. Let $(y_i)_{i=1}^\infty$ be the sequence obtained by merging all $N$ sequences $(x_{j,i})_{i=1}^\infty$ and sorting the resulting set of values in nondecreasing order. Then, it holds that
	\begin{equation*}
		A^\frac 1N\leq\liminf_{i\to\infty}y_i^\frac 1i\leq\limsup_{i\to\infty}y_i^\frac 1i\leq B^\frac 1N.
	\end{equation*}
\end{lemma}

\begin{proof}
	For each $j=1,\ldots,N$, we define the function $f_j(t)\in\mathbb{N}$, where $t>0$, as the number of $x_{j,i}$'s such that $x_{j,i}\leq t$. More precisely, for any $t>0$, the value $f_j(t)$ is the unique nonnegative integer such that
	\begin{equation*}
		x_{j,i}\leq t,\text{ for all }1\leq i\leq f_j(t)
		\quad\text{and}\quad
		x_{j,i}> t,\text{ for all }i>f_j(t).
	\end{equation*}
	In particular, if $c_Bi^RB^i\leq t$, then $f_j(t)\geq i$. Therefore, if $c_Bi^RB^i\leq t<c_B(i+1)^RB^{i+1}$, using that $\log(i+1)\leq\sqrt i$ for all $i\geq 0$, we conclude that
	\begin{equation*}
		f_j(t)\geq i>\frac{\log t-\log c_B}{\log B}-1-\frac{R}{\log B}\sqrt{i}\geq \frac{\log t-\log c_B}{\log B}-1-\frac{R}{\log B}f_j(t)^\frac 12.
	\end{equation*}
	Similarly, noticing that, if $t<c_AA^i$, then $f_j(t)<i$, we conclude that
	\begin{equation*}
		f_j(t)< i\leq\frac{\log t-\log c_A}{\log A}+1,
	\end{equation*}
	whenever $c_AA^{i-1}\leq t<c_AA^{i}$.
	
	Now, consider any large $t>0$ such that $y_i\leq t<y_{i+1}$, for some large $i$. It must then hold that
	\begin{equation*}
		i=\sum_{j=1}^Nf_j(t),
	\end{equation*}
	whence
	\begin{equation*}
		N\left(\frac{\log t-\log c_B}{\log B}-1-\frac{R}{\log B}\sqrt{i}\right)<i<N\left(\frac{\log t-\log c_A}{\log A}+1\right).
	\end{equation*}
	In particular, setting $t=y_i$ yields that
	\begin{equation*}
		c_AA^{\frac iN-1}<y_i<c_Be^{R\sqrt i}B^{\frac iN+1},
	\end{equation*}
	which implies
	\begin{equation*}
		A^\frac 1N=\lim_{i\to\infty} c_A^\frac 1iA^{\frac 1N-\frac 1i}
		\leq\liminf_{i\to\infty}y_i^\frac 1i\leq\limsup_{i\to\infty}y_i^\frac 1i
		\leq \lim_{i\to\infty} c_B^\frac 1ie^{\frac R{\sqrt i}}B^{\frac 1N+\frac 1i}=B^\frac 1N,
	\end{equation*}
	thereby completing the proof of the lemma.
\end{proof}

\section{Omitted proofs of Section~\ref{sec:noisy}}

\begin{proof}[Proof of Theorem~\ref{thm:cyclic.upper}]
We will reduce {\sc MinCyclic} to a problem studied in~\cite{RivestMKWS80}. The latter is defined as a game between two players. Player $A$ guesses an integer, say $x$, in the range $[0, \ldots ,m-1]$. Player $B$ knows $m$, and must find $x$ using as few comparison queries as possible, in the presence of errors. Specifically, $A$ can give erroneous responses to at most $H$ queries, for some $H$ that is known to player $B$. Let $Q(m,H)$ denote the number of comparison queries that are sufficient for player $B$ to find $x$. Rivest et al.~\cite{RivestMKWS80}, presented an algorithm, which we call \wa, with $Q(m,H) \leq \min \{k | 2^{k-H} \geq m \cdot \mch{k-H}{H}\}$. In other words, for any $k$ that certifies 
\begin{equation}
    2^{k-H} \geq m \cdot \mch{k-H}{H}, \label{eq:rivest}
\end{equation} 
it is possible to find $x$ using $k$ queries. Note that it must be that $H\leq k/2$.

Given an instance of {\sc MinCyclic}, we create an instance of the above problem, with 
$m = 2^{k-H}/\mch{k-H}{H}$ (note that $H$ is the same for both instances). Partition the interval $[0, \ldots, n-1]$ into $m$ disjoint subintervals, each of length at most $\lceil n/m \rceil$. Let $x$ being the index in $[0, \ldots ,n-1]$ for which $A[x]=0$, and let $I_x$ denote the interval that contains $x$. 
A query $q$ of \wa that asks ``is $I_x \leq b$ for some $b \in \{0,m-1\}$?" 
translates to query $\mu(x)$ in the {\sc MinCyclic} instance that asks ``is $x \leq f(b)$?", where $f(b)$ is the largest value in the interval $I_b$. Note that the answer to $q$ is `yes' if and only if the answer to $\mu(q)$ is `yes'. The response to $\mu(q)$ is then given to \wa  which updates its state and proceeds with the next query. For $m$ defined as above,~\eqref{eq:rivest} holds and using \wa, we can find $I_x$ using $k$ queries.  
Subsequently, we return the largest integer $f(x)$ in $I_x$. 
Given that $x$ is in $I_x$ and the length of the intervals is at least  $\lceil n/m \rceil$, we conclude that the returned index $j$ is such that $A[j]=\leq \lceil n/m \rceil = \lceil n \mch{k-H}{H}/2^{k-H} \rceil $.
\end{proof}

\begin{proof}[Proof of Theorem~\ref{thm:cyclic.lower}]
Consider the following problem that can be described as a game between two players $A$ and $B$. $A$ guesses a value $r$ in the continuous interval $[0,n)$, for some fixed $n$ (for the purpose of the proof, we can think of $n$ as sufficiently large). Player $B$ knows $n$ but not $r$, and $B$ asks $k$ (subset) queries to $A$, $H$ of which can receive erroneous responses, for some $H \leq k$ known to both players. After receiving the responses to the $k$ queries, player $B$ returns a number $r' \in [0,n)$. The objective for $B$ is to minimize $|r'-r|$. We will denote this problem as {\sc Game}.
In~\cite{RivestMKWS80}, it is proved that, for all values of $H$, player $A$ can respond in a way that guarantee that $|r-r'| \geq n \mch{k}{H}/2^{k}$.

We now show a reduction from {\sc Game} to \search that will help us establish the lower bound. 
Suppose, by way of contradiction, that there exists an algorithm, say ALG for \search that returns an element $e$ with $\pi(e) < \lfloor n\mch{k}{H}/2^{k} \rfloor - 1$. We devise a strategy for player $B$ in {\sc Game} based on ALG. Given values of $r$, $k$ and $H$ that define an instance $G$ of {\sc Game} over the continuous interval $I = [0,n)$, create an instance $S$ of \search on a space $A$ of $n$ elements, with the same values of $k$ and $H$. Consider a bijective mapping $\beta$ that maps an element of rank $i$ in $A$ ($i\in \{0, \ldots, n-1\}$) to an interval $\beta(i) = [i,i+1)$ in $I$. 
Similarly, define a bijective mapping $\mu$ between queries asked for $G$ and those asked for $S$. Any range $[i,j]$ of indices that is a part of a subset query $q$ asked for $S$ is mapped to an interval $[i,j+1)$ in the query $\mu(q)$ asked for $G$. 
Let $r$ denote the searched value in $G$ and let $x$ denote an index of $A$ such that $r$ belongs to $\beta(x)$.
To search for $r$, we consider queries that ALG asks for $S$ and for any such query $q$, we ask $\mu(q)$ for $G$. The response to $\mu(q)$ is then given to ALG so that it can update its state and ask its next query.

Recall that we supposed ALG outputs an element $e$ such that $\pi(e) < \lfloor n\mch{k}{H}/2^{k} \rfloor -1$, and that $\beta(\pi(e)) = [\pi(e),\pi(e)+1)$. As an output for $G$, we return $r' = \pi(e)+1$ as the answer for $G$. Note that there are exactly $e+1$ intervals from the range of $\beta$ that lie between $r$ and $r'$ in $I$. That is $|r-r'| < \lfloor n\mch{k}{H}/2^{k} -1\rfloor+1 \leq n(\mch{k}{H}/2^{k}$. This, however, contradicts the result of~\cite{RivestMKWS80}. 
\end{proof}

\begin{proof}[Proof of Theorem~\ref{thm:multi-search.alpha.faulty}]
Consider a schedule $X$ for this problem. For $j \in [0,\ldots p-1]$, define $l_X(t,j)$ as the length of the largest contract in $S$ 
that has completed in processor $j$ by time $t$. We also define by $\ell_{X,\phi}(t)$ as the $(\phi+1)$-largest length in the set 
$\{l_X(t,j)\}_{j=0}^{p-1}$. 

Following the notation of~\cite{aaai06:contracts}, we denote each contract $c_j$ in $X$ as a pair of the form  $(T_j,D_j)$, where $T_j$ is the start time of $c_j$, and $D_j$ its length (as we will see, the specific processor to which the contract is assigned will not be significant for our analysis). We also define $d_j$ to be equal to $\ell_{X,\phi}(T_j+D_j-\epsilon)$, 
for $\epsilon \rightarrow 0$. In words, $d_j$ is the longest contract length that has been completed right before $c_j$ is about to terminate, assuming a worst-case scenario in which $\phi$ processors have been faulty, and they also happened to be the processors that have completed the longest contracts in the schedule by the said time: we call this contract length the {\em $(\phi+1)$ length relative to $D_j$}.

Recall that $\bar{X}$ denotes the sequence of all contract lengths in $X$, in non-decreasing order. 
Hence, each contract in $X$ is mapped via its length to an element of this sequence (breaking ties arbitrarily).

Fix a time, say $t$, at which a contract $c_{j_0}=(T_{j_0},D_{j_0})$ terminates, say on processor $0$ i.e., $t=T_{j_0}+D_{j_0}$. For all $m \in [1, p-1]$, let $c_{j_m}=(T_{j_m},D_{j_m})$ denote the longest contract length that has completed on processor $m$ by time $t$. For every $m \in [0,p-1]$, define $I_m$ as the set of indices in ${\mathbb N}$ such that $i \in I_m$ if and only if the contract of length $x_i$ has completed by time $t$ in processor $m$. From the definition of the acceleration ratio we have that 
\[
\acc(X) \geq \frac{\sum_{i \in I_m} x_i}{d_{j_m}}, \qquad \ \textrm{ for all $m \in [0, p-1]$.}
\]
Therefore,
\[
\acc(X) \geq \max_{0 \leq m \leq p-1} \frac{\sum_{i \in I_m} x_i}{d_{j_m}},
\]
and using the property $\max\{ a/b, c/d \} \geq \frac{a+b}{c+d}$, for all $a,b,c,d>0$, we obtain that
\begin{equation}
\acc(X) \geq  \frac{\sum_{m=0}^{p-1}\sum_{i \in I_m} x_i}{\sum_{m=0}^{p-1} d_{j_m}}.
\label{eq:summand}
\end{equation}
Next, we will bound the numerator of the fraction in~\eqref{eq:summand} from below, and its denominator from above. We begin with a useful observation: we can assume, without loss of generality, that by time $t$ (defined earlier), every contract of length $d_{j_0}$ or smaller has completed its execution. This follows from the definition of $d_{j_0}$: if $X$ completed a contract of length at most $d_{j_0}$ later than time $t$, then one could simply ``remove'' this contract from $X$, and obtain a schedule of no worse acceleration ratio (in other words, such a contract is useless, and one can derive a schedule of no larger acceleration ratio than $X$ that does not contain it). 

Using the above observation, it follows that the numerator in~\eqref{eq:summand} includes, as summands, all contracts of length at most $d_{j_0}$, as well as at least $\phi+1$ contracts that are at least as large as $d_{j_0}$. Let $q$ denote an index such that $d_{j_0}=\bar{x}_{q}$, then we have that 
\[
\sum_{m=0}^{p-1}\sum_{i \in I_m} x_i \geq \sum_{i=0}^{q+\phi+1} \bar{x}_i.
\]
We now show how to upper-bound the denominator, using the monotonicity implied in the definition of 
the $(\phi+1)$ length relative to a given contract length. Since $c_{j_0}$ is completed no earlier than any other contract 
$c_{j_m}$, with $m \in [1, p-1]$, we have that $d_{j_m} \leq d_{j_0}$. It thus follows
that
\[
\sum_{m=0}^{p-1} d_{j_m} \leq \sum_{i=q}^{q-(p-1)} \bar{x}_{i}.
\]
Combining the two bounds, it follows that 
\[
\acc(S) \geq \sup_{0\leq q<\infty} \frac{\sum_{i=0}^{q+\phi+1} \bar{x}_i}{\sum_{i=q}^{q-(p-1)} \bar{x}_{i}}.
\]
Define now the functional $F_{q}(\bar{X})=\frac{\sum_{i=0}^{q+\phi+1} \bar{x}_i}{\sum_{i=q}^{q-(p-1)} \bar{x}_{i}}$, for every $q$. The functional satisfies the conditions (1)-(5) of 
Theorem~\ref{thm:gal} (see Example 7.3 in~\cite{searchgames}). Moreover, $\sup_{n \geq 0} \bar{x}_{n+1}/\bar{x}_n <\infty$, otherwise an infinite contract would be scheduled in some processor, which, in turn, would render the corresponding processor ``useless'', since this contract would never complete. Last, we note that the condition $\alpha_X>0$ is indeed satisfied, from Theorem~\ref{thm:zetas} and Corollary~\ref{cor:merge}.

By applying Gal's Theorem, it follows that 
\[
 \acc(X) \geq 
 \sup_{0 \leq q<\infty} 
 \frac {    \sum_{i=0}^{q+\phi+1}  \alpha_{\bar{X}}^i   } 
{\sum_{i=q}^{q-(p-1)} \alpha_{\bar{X}}^i }.
\]
If $\alpha_{\bar{X}}\leq 1$, then it is easy to show that the above expression shows that $\acc(X)=\infty$; see, e.g.~\cite{aaai06:contracts}. Otherwise, i.e., if $\alpha_{\bar{X}}>1$, after some simple calculations along the lines of~\cite{aaai06:contracts} we arrive at the desired result.
\end{proof}

\section{Omitted proofs of other sections}

\begin{proof}[Proof of Theorem~\ref{thm:rft}]
We first show the lower bound. Let $X$ denote a $p$-processor schedule for {\sc Rft}($r,p,f$). 
This schedule is defined by a set of $p$ (single-processor) schedules, say $\{X_0, \ldots ,X_{p-1}\}$, which, by the definition of the problem, must be $r$-robust. We can apply Theorem~\ref{thm:multi-search.alpha.faulty} with $\phi=f$, which yields
\[
\acc(X) \geq \frac{\alpha_{\bar{X}}^{p+1+f}}{\alpha_{\bar{X}}^p-1}.
\]
Moreover, since every schedule in the set that defines $X$ must be $r$-robust, 
from Corollary~\ref{cor:merge} we have that 
\[
\alpha_{\bar X} \in [\zeta_{1, r}^{1/p}, \zeta_{2,r}^{1/p}].
\]
Therefore, $X$ must be such that
\begin{equation}
\acc(X) \geq  \frac{\alpha_{\bar{X}}^{p+1+f}}{\alpha_{\bar{X}}^p-1} \quad \textrm{subject to} \quad  \zeta_{1, r}^{1/p} \leq \alpha_{\bar{X}} \leq \zeta_{2, r}^{1/p}
\label{eq:lb.robust.ft}
\end{equation}
We now show a schedule for the problem whose acceleration ratio matches the above lower bound, and is thus optimal. Let $X$ be defined by the set ${\cal X}_{b,p}$, for some $b>1$ that will be specified later. 

Consider an interruption $T$ that occurs right before the contract of length $b^{jp+l}$ terminates, for some $j \in \mathbb{N}$, and $l \in [0, p-1]$. Then, from the definition of ${\cal X}_{b,p}$, a contract of length $b^{jp+l-f-1}$ has completed, even if $f$ processor faults have occurred. 
We have 
\[
\frac{T}{b^{jp+l-f-1}}=\frac{\sum_{i=0}^j b^{ip+l}}{b^{jp+l-f-1}} \leq \frac{b^{p+f+1}}{b^p-1}.
\] 

We require that each schedule $X_i$ in ${\cal X}_{b,p}$ is $r$-robust, which is enforced as long as
$\frac{b^{2p}}{b^p-1} \leq r$, or equivalently, $b^{p} \in [\zeta_{1,r}^{1/p}, \zeta_{2,r}^{1/p}]$. Summarizing, the acceleration ratio of the schedule $X$ is at most
\[
\frac{b^{p+f+1}}{b^p-1} \quad \textrm{subject to} \ b \in [\zeta_{1,r}^{1/p}, \zeta_{2,r}^{1/p}].
\]
Choosing a value of $b$ that optimizes the above expression will yield an optimal schedule, by directly comparing the expression to the lower bound given by~\eqref{eq:lb.robust.ft}. 
\end{proof}

\end{document}